\documentclass[twoside]{article}

\usepackage[usenames,dvipsnames]{color}
\usepackage{amsmath}
\usepackage{amsfonts}
\usepackage{amssymb}
\usepackage{amsthm}
\usepackage{amsmath}
\usepackage{amsthm}
\usepackage{hyperref}
\usepackage{latexsym}
\usepackage[T1]{fontenc}
\usepackage{color}

\def\pp{\vskip3ex}

\newcommand{\eq}[1]{\begin{equation}
                     \begin{split} #1 \end{split}
                     \end{equation}}

\newcommand{\R}{\mathbb{R}}

\newcommand{\RR}{{\mathcal{R}}}

\newtheorem{theorem}{Theorem}
\newtheorem{definition}{Definition}
\newtheorem{rem}[definition]{Remark}

\def\DFT{double field theory\ }
\def\DFTp{double field theory}
\begin{document}
\thispagestyle{empty}
\vspace*{-1.5cm}

\begin{flushright}
  {
  ITP-UH-07/14\\
  
  }
\end{flushright}

\vspace{1.5cm}

%%%%%%%%%%%%%%%%%%%%%%%%%%%%%%%%%%%%%%%%%%%%%%%
%%%%%%%%%%%%%%%%%%%%%%%%%%%%%%%%%%%%%%%%%%%%%%%

\begin{center}
{\huge
Even symplectic supermanifolds   \\
and double field theory
}
\end{center}

%%%%%%%%%%%%%%%%%%%%%%%%%%%%%%%%%%%%%%%%%%%%%%%
%%%%%%%%%%%%%%%%%%%%%%%%%%%%%%%%%%%%%%%%%%%%%%%

\vspace{0.4cm}

\begin{center}
  Andreas Deser$^{1}$ and Jim Stasheff$^{2}$
\end{center}

%%%%%%%%%%%%%%%%%%%%%%%%%%%%%%%%%%%%%%%%%%%%%%%
%%%%%%%%%%%%%%%%%%%%%%%%%%%%%%%%%%%%%%%%%%%%%%%

\vspace{0.4cm}

\begin{center} 
\small{\emph{$^{1}$ Institut f\"ur Theoretische Physik } and \emph{Riemann Center for Geometry and Physics}, \\ 
   {\em Leibniz Universit\"at Hannover}\\
%{\em Appelstra\ss e 2, 30167 Hannover, Germany}}\\
}
\small{Email: {\tt andreas.deser@itp.uni-hannover.de} }

\vspace{0.25cm}

\small{\emph{$^{2}$ UNC-CH }and \emph{University of Pennsylvania}}\\
\small{Email: {\tt jds@math.upenn.edu}}

\vspace{0.25cm}

\end{center} 

\vspace{0.4cm}

%%%%%%%%%%%%%%%%%%%%%%%%%%%%%%%%%%%%%%%%%%%%%%%
%%%%%%%%%%%%%%%%%%%%%%%%%%%%%%%%%%%%%%%%%%%%%%%
\begin{center}
\today
\pp
%\trd{Kirill-I've inserted comments in red where I need your help to say things correctly.}
\end{center}

\begin{abstract}
Over many decades, the word ``\emph{double}'' has appeared in various contexts, at times seemingly unrelated\footnote{Compare the over use of \emph{twisting}.}. Several have some relation  to mathe\-ma\-ti\-cal physics. Recently, this has become particularly strking in DFT (\DFTp).
\pp
Two `doubles' that are particularly relevant are
\begin{itemize}
\item double vector  bundles and
\item Drinfel'd doubles.
\end{itemize}
The original Drinfel'd double occurred in the contexts of quantum groups \cite{drinfeld:qgroups} and of Lie bialgebras \cite{drinfeld:Poisson}.

Quoting T. Voronov \cite{thedya:Q&Mac}:
\pp

\begin{quote}
Double Lie algebroids arose in the works on double Lie
grou\-po\-ids~\cite{mackenzie:secondorder1},
 \cite{mackenzie:secondorder2} and in connection with an analog for
Lie bialgebroids of the classical Drinfel'd double of Lie
bialgebras~\cite{mackenzie:doublealg}, %\cite{mackenzie:doublealg2},
\cite{mackenzie:drinfeld}.\dots
%Recall that Drinfeld's \textit{classical double} of a Lie bialgebra
%is again a Lie bialgebra with ``good'' properties. 
 %Three constructions of a `double' for Lie bialgebroids as an analog of Drinfeld's
%construction  have been suggested. 
Suppose $(A,A^*)$
is a Lie bialgebroid over a base $M$.\dots
Mackenzie in~\cite{mackenzie:secondorder1}, \cite{mackenzie:drinfeld},
\cite{mackenzie:notions} and Roytenberg in~\cite{Deethesis}
suggested two different constructions based on the cotangent bundles
$T^*A$ and $T^*\Pi A$, respectively. Here $\Pi$ is the fibre-wise parity reversal functor.
\end{quote}

Although the approaches of Roytenberg and of Mackenzie look very
different, Voronov  establishes their equivalence.
We have found Roytenberg's version to be quite congenial with our attempt
to interpret the gauge algebra of \DFT\  in terms of 
 Poisson brackets on a suitable generalized Drinfel'd double.
This double  of a Lie bialgebroid $(A,A^*)$ provides a framework 
to describe the differentials of $A$ and $A^*$ on an equal footing as Hamiltonian functions on an even symplectic supermanifold. 
A special choice of momenta explicates the double coordinates of DFT\ and shows their relation to the strong constraint determining the physical fields of double field theory.

\end{abstract}

%%%%%%%%%%%%%%%%%%%%%%%%%%%%%%%%%%%%%%%%%%%%%%%
%%%%%%%%%%%%%%%%%%%%%%%%%%%%%%%%%%%%%%%%%%%%%%%

\clearpage

\tableofcontents

%%%%%%%%%%%%%%%%%%%%%%%%%%%%%%%%%%%%%%%%%%%%%%%
%%%%%%%%%%%%%%%%%%%%%%%%%%%%%%%%%%%%%%%%%%%%%%%

\section{Introduction}
Closed string theory on a toroidal target space has the peculiarity of ex\-hi\-bi\-ting two sets of momentum-type variables: One canonically conjugate to the standard center of mass position coordinates and another describing the winding degrees of freedom of a closed string around a compact cycle on the target space. Associating to the latter canonically conjugate coordinates, one is lead to \emph{double} the number of %(compact)
 configuration space coordinates. 

The mass spectrum of a quantized closed string on a toroidal geometry of length scale $R$ enjoys a distinctive symmetry, called \emph{T-duality} which inverts the scale $R \mapsto l_s^2/R$ ($l_s$ is called the fundamental string length) and simultaneously \emph{exchanges} standard momentum and winding quantum numbers. As a consequence, on the classical level the aim for a space-time description of the massless sector of a closed string (incorporating the metric $G_{ij}$ a two-form $B_{ij}$ and a scalar dilaton $\Phi$) having a manifest T-duality symmetry requires a field theory being covariant under the exchange of the two sets of configuration space coordinates, or more generally under $O(d,d)$-transformations.

Historically,
double field theory (DFT) \cite{Siegel:1993th,Hull:2009mi,Hohm:2010jy,Hohm:2010pp} was a proposal to incorporate these transformations as a symmetry of a field  theory defined on a \emph{double configuration space} with coordinates $(x^i,\tilde x_j)$. 
%Naively, doubling refers to doubling coordinates, meaning local coordinates.

Prior to doubling,  there was a need to find a common setting for both  diffeomorphisms   
and $B$-field gauge transformations.
This is already the  basic idea of  \emph{generalized geometry}, 
%by Hitchin and Gualtieri \cite{Hitchin:2004ut,Gualtieri:2003dx} is to
combining vector fields and one-forms into a single object. Formally, on a
manifold $M,$ one introduces a \emph{generalized tangent
bundle} $E$ which is a particular extension of $T$ by $T^*$.  
%\emph{Locally}, the `fields' depend on local coordinates $(\tilde x_\mu, \tilde y_m, x^\mu, y^m)$ of some space.
This construction is a typical example of a \emph{Courant algebroid} \cite{0885.58030} arising from a Lie bialgebroid \cite{Mackenzie:GT,kirill&ping}. $B$-field transformations are automorphisms of the Courant bracket for the case of $dB=0$, whereas for general $B$-fields, gerbes are used.

For application to DFT, the use of doubles in the Drinfel'd sense suggests itself. Consider $(TM,T^*M)$ or more generally $(A,A^*)$
as a Lie bialgebroid over a base $M$.
Mackenzie in~\cite{mackenzie:drinfeld}, \cite{kirill:Crelle}, and Roytenberg in~\cite{Deethesis}
suggested two different constructions based on the cotangent bundles
$T^*A$ and $T^*\Pi A$, respectively. Here $\Pi$ is the fibre-wise parity reversal functor.
Although the approaches of Roytenberg and of Mackenzie look very
different, T. Voronov \cite{thedya:Q&Mac}  establishes their equivalence.
We have found an elegant and unifying language for Lie bialgebroids, Courant algebroids and their exterior algebras using even symplectic supermanifolds as given by Roytenberg in \cite{Deethesis}, see also \cite{Roytenberg:2001am,Roytenberg:2002nu}; a translation into Mackenzie's formalism might serve as well, but physicists seem comfortable with supermanifolds.

%\trd{Should there be a comment about doing this without superizing? perhaps a very strong comment?}

Whereas the relation of the ``supermathematical'' viewpoint to generalized geometry and its physical applications is clear and precise, it is natural to ask about its relation to double field theory. Conversely, a precise definition of double fields and an understanding of the gauge algebra of DFT could lead to new insights on the physics of double fields, e.g. the incorporation of three-form fluxes and three-tensors in a mathematically precise way. 

In the following, we start by reviewing the basic mathematical language of Lie bialgebroids, Courant algebroids and doubles.
Then we introduce even symplectic supermanifolds as used in \cite{Deethesis} as well as basic facts of double field theory, especially the algebra of its gauge transformations. These sections are intended to present only the minimal amount of material needed to follow the rest of the exposition. The section following these more introductory parts contains the main result:  an interpretation of double fields in terms of functions on a suitable
 double of a Lie bialgebroid. 
 
 As a result, we will derive the C-bracket of DFT by using the Poisson algebra on the underlying supermanifold
 and reveal it as a Courant bracket; only the formula seemed to indicate it was something more unusual. As applications, we will elaborate on the strong constraint of DFT and on a special projection of it to standard generalized geometry by dropping the winding coordinates. We conclude with an outlook on possible advantages of the established formalism.

\section{The language: Lie bialgebroids and various `Drinfel'd' doubles}
\label{sec-language}

We begin by reviewing the subject of Lie bialgebroids  \cite{Mackenzie:GT} and then turn to the ``supermathematical'' viewpoint, closely following the clear and detailed exposition of the subject given in
Roytenberg's thesis \cite{Deethesis}.  Starting with a Lie algebroid $A$ and its dual $A^*$ and  applying the parity reversal functor $\Pi$ to fibres,
 it is possible to recast the Lie bialgebroid condition in terms of the canonical even symplectic structure on the cotangent bundle of $\Pi A$. 

\subsection{Lie bialgebroids and Courant algebroids}
Lie algebroids are a common generalization of the tangent bundle and Lie algebras. They are vector bundles equipped with a Lie bracket on  the sections of the bundle and should not  be thought of as ``bundles of Lie algebras''.
 Besides the use in Poisson-geometry, they have been used recently in physics, especially in the application of generalized geometry and string theory \cite{Vaisman:2012ke,Blumenhagen:2012nt,Blumenhagen:2013aia}. The precise definition is 
\begin{definition}
A vector bundle $A \rightarrow M$, equipped with a skew-symmetric bracket $[\cdot,\cdot]_A$ on the space of sections $\Gamma(A)$ satisfying the Jacobi identity, and a bundle homomorphism $a : A \rightarrow TM$ (termed anchor) is called \emph{Lie algebroid} if the following Leibniz rule holds:
\eq{
[X,fY]_A = \, f[X,Y]_A + \Bigl(a(X)f\Bigr)Y\;,
}
for $X,Y \in \Gamma(A)$ and $f \in {\cal C}^{\infty}(M)$.
\end{definition}

As a consequence, on the exterior algebra of sections in the dual bundle $\Gamma(\wedge^\bullet A^*),$ it is possible to define the Chevalley-Eilenberg differential:
\eq{
d_A \omega(X_0,&\dots,X_k) =\;  \sum_{i=0}^k \,(-1)^i a(X_i)\Bigl(\omega(X_0,\dots,\hat X_i,\dots X_k)\Bigr) \\
&+ \underset{0\leq i < j \leq k}\sum\, (-1)^{i+j}\omega([X_i,X_j]_A,X_0,\dots,\hat X_i, \dots,\hat X_j,\dots X_k) \;,
}
for $\omega \in \Gamma(\wedge^k A^*)$. The Jacobi identity for the bracket on $A$ is used to prove that $d_A^2 = 0$. It turns out that differential geometric notions known from the tangent bundle case are generalizable to Lie algebroids. Examples are Lie derivatives, defined in the standard way as the graded commutator of $d_A$ and the insertion map $i_X(\omega) = \omega(X,\cdots)$, i.e. for sections $X$ in $A$,  $L^A_X = [d_A,i_X] = d_A\circ i_X + i_X \circ d_A$.

Clearly, the tangent bundle $(TM,[\cdot,\cdot],\textrm{id})$ is a Lie algebroid, where the bracket is the Lie bracket or more generally the Schouten bracket of polyvector fields. The corresponding differential on the exterior algebra of its dual is the standard de Rham differential. If $M$ in addition is Poisson, also $T^*M$ can be equipped with a bracket (called Koszul-Schouten bracket \cite{zbMATH03996707}) and an anchor in terms of the Poisson tensor, turning it into a Lie algebroid. The de Rham differential in addition is a derivation of the Koszul-Schouten bracket, which is a motivation for the following definition~\cite{Kosmann-Schwarzbach:1995}.
\begin{definition}
A pair $(A,A^*)$ of a Lie algebroid $A$ and its dual is called \emph{Lie bialgebroid} if the differential $d_A$ on $\Gamma(\wedge^\bullet A^*)$ is a derivation of the bracket $[\cdot,\cdot]_{A^*}$ on $\Gamma(\wedge^\bullet A^*)$.
\end{definition}
One can show
 \cite{kirill&ping,Kosmann-Schwarzbach:1995}
that if $(A,A^*)$ is a Lie bialgebroid, then also $(A^*,A)$ has this property.
Lie bialgebroids are special examples of \emph{Courant algebroids}. We start by giving a definition which is equivalent to the original definition used in \cite{zbMATH00004959}\footnote{We refer the reader to \cite{yks:short}
 for a proof of the equivalence of the definitions.}.
\begin{definition}
A vector bundle $E \rightarrow M$ equipped with a nondegenerate bilinear form $\langle\cdot,\cdot\rangle$ and a bilinear operation $\circ$ on sections in $E$, together with a bundle map $\rho : E \rightarrow TM$ is called \emph{Courant algebroid} if the following properties hold:
\begin{itemize}
\item $s_1 \circ (s_2 \circ s_3) = (s_1\circ s_2)\circ s_3 + s_2 \circ (s_1 \circ s_3)$ for $s_i \in \Gamma(E)$, 
\item $\rho(s_1\circ s_2) = [\rho(s_1),\rho(s_2)]$ for $s_i \in \Gamma(E)$,
\item $s_1 \circ f\,s_2 = f(s_1\circ s_2) +\Bigl(\rho(s_1)f\Bigr)s_2$ for $s_i\in \Gamma(E), f \in {\cal C}^\infty(E)$, 
\item $s\circ s = \tfrac{1}{2}{\cal D}\langle s,s\rangle$ for $s \in \Gamma(E)$ and ${\cal D}: {\cal C}^\infty(M) \rightarrow \Gamma(E)$ defined by $\langle {\cal D}f,s\rangle = \rho(s)f$. 
\item $\rho(e)\langle s_1,s_2\rangle = \langle e \circ s_1,s_2\rangle + \langle s_1,e\circ s_2 \rangle$ for $e,s_i \in \Gamma(E)$\;.
\end{itemize}
\end{definition}
It is important to note that the map $\circ$, although called the \emph{Dorfman} bracket,  is not skew-symmetric in general, so we do not want to call it a bracket. However, it is related to the standard Courant bracket of two sections $s_1,s_2$ by
\eq{
[s_1,s_2] = \; \frac{1}{2}(s_1\circ s_2 - s_2 \circ s_1)\;.
}
To see the relation to Lie bialgebroids, let $(A,[\cdot,\cdot]_A,a)$ be a Lie algebroid and $(A^*,[\cdot,\cdot]_{A^*},a_*)$ be its dual such that the pair $(A,A^*)$ is a Lie bialgebroid. Let further $E = A \oplus A^*$. We denote sections in $E$ by $X+\eta, Y+\omega \in \Gamma(E)$ and define the operation $\circ$ and the bilinear form $\langle \cdot, \cdot \rangle$ by 
\eq{
\langle X+\eta, Y+\omega \rangle =&\; \eta(Y) + \omega(X) \;, \\
(X+\eta)\circ (Y+\omega) &=\;\Bigl([X,Y]_A + L^{A^*}_\eta Y - i_\omega d_{A^*}X\Bigr) \\
&\;+ \Bigl([\eta,\omega]_{A^*} + L^A_X \omega - i_Y d_A \eta\Bigr)\;.
}
Together with the bundle map determined by the two anchors, $\rho(X+\eta) = a(X) + a_*(\eta)$, it is possible to prove the following\footnote{We refer to \cite{lwx} for a detailed proof.}

%\trd{have changed the citaation to LWX}

%\trd{Instead of the footnote, suggest including the citation and Theorem number there after begintheorem}

%\trd{will change the following to the earliest reference}

\begin{theorem}[\cite{0885.58030}]
If $(A,A^*)$ is a Lie bialgebroid, then $(E,\langle\cdot,\cdot\rangle, \circ, \rho)$ is a Courant algebroid.
\end{theorem}
The corresponding Courant bracket is given by antisymmetrizing the $\circ$-operation. As similar brackets play a role in double field theory, we give their explicit form:
\eq{
[X+\eta,Y+\omega] = &\;[X,Y]_A + L_\eta^{A^*}Y - L_\omega^{A^*}X - \frac{1}{2}d_{A^*}\Bigl(Y(\eta) - X(\omega)\Bigr)  \\
&+ [\eta,\omega]_{A^*} + L_X^{A}\omega - L_Y ^A\eta + \frac{1}{2} d_A\Bigl(\eta(Y) - \omega(X)\Bigr)\;.
}
Again, the most direct example is the case of a Poisson manifold $M$, where $A = TM$. In the previous expression, we then have the standard Lie bracket for $[\cdot,\cdot]_A$ and the Koszul-Schouten bracket for $[\cdot,\cdot]_{A^*}$ \cite{zbMATH03996707}. The differentials are the de Rham differential on $\Gamma(\wedge^\bullet A^*)$ and the Poisson differential (given in terms of the Schouten-Nijenhuis bracket \cite{0023.17002,nijenhuis:jacobi} and the Poisson tensor) on $\Gamma(\wedge^\bullet A)$.

\subsection{Reformulation in terms of supermanifolds}
The structures introduced in the previous section can be reformulated in the language of supermanifolds. One of the advantages of this language is to get a meaningful notion of the ``sum'' of the two differentials $d_A$ and $d_{A^*}$. Note that  the two operators act on different spaces but we can combine them on a single space by interpreting them as functions $h_{d_A}$ and $h_{d_{A^*}}$ related by a Legendre transform $L$. One then has $h_{d_A} + L^*h_{d_{A^*}}$ as a well defined sum.

\subsubsection{Parity reversed Lie algebroids and their cotangent bundles}
A detailed introduction to supermanifolds would go far beyond the scope of this paper and is not needed in later chapters, so we refer the reader to the literature on the subject, e.g. \cite{Cattaneo:2010re} . We again closely follow the exposition of \cite{Deethesis} and begin with an informal description of elementary supergeometric notions. 

Let $V= V_0 \oplus V_1$ be a finite dimensional $\mathbb{Z}_2$-graded vector space.\footnote{The grading is often called parity.}. The elements of $V_0$ and $V_1$ are called even and odd elements, respectively. We define \emph{parity reversion} $\Pi$ by 
\eq{
(\Pi V)_0 = V_1 \quad \textrm{and}\quad (\Pi V)_1 = V_0\;.
}
The dimension of $V$ is often denoted by $\textrm{dim}V = (\textrm{dim} V_0 | \textrm{dim} V_1)$.
If $V$ is a finite dimensional even vector space, then polynomial functions on $\Pi V$ can be identified with elements in the exterior algebra of the dual of $V$, i.e. 
\eq{
\textrm{Pol}^\bullet (\Pi V) \simeq \wedge ^\bullet V^*\;.
}
Considering the super-vectorspace $\mathbb{R}^{(n|m)}$ and an open subset $U_0 \subset \mathbb{R}^n$, polynomials in the odd variables with coefficients depending smoothly on the even variables in $U_0$ are smooth functions on the superdomain $U^{(n|m)}$. They form the supercommutative algebra ${\cal C}^\infty(U) \otimes \wedge^\bullet(\mathbb{R}^m)^*$. The idea of a \emph{supermanifold} is to have coordinate patches given by superdomains of the form $U^{(n|m)}$. We are not going to introduce locally ringed spaces to make the notion of a supermanifold precise, but  only mention the possibility of defining symplectic (and Poisson) supermanifolds. A symplectic form $\omega$ on a supermanifold $Q$ with even coordinates $x^i$ and odd coordinates $\xi^m$, is a closed and non-degenerate two-form\footnote{Which can be seen as a quadratic function on $\Pi TQ$ as we explained before.} 
\eq{
\omega = \frac{1}{2} \omega_{ij} dx^i\wedge dx^j + \omega_{im} dx^i\wedge d\xi^m + \omega_{mn} d\xi^m \wedge d\xi^n\;,
}
with skew-symmetric part $\omega_{ij}$ and \emph{symmetric} part $\omega_{mn}$. Poisson brackets are defined similarly to the purely even case, but having in mind the graded nature of the $\xi^m$, e.g. 
\eq{
\{\xi^m,\xi^n\} = \{\xi^n, \xi^m\}\;.
}

\def\Oplus{\bigoplus}

\begin{rem}
There is an alternate point of view using a $\mathbb Z$-grading. For a  $\mathbb Z$-graded module $V = \Oplus V_n$ over a commutative ring $\RR$, denote by $sV$ the graded module 
$sV=\Oplus (sV)_{n+1}$ where $ (sV)_{n+1}$ is isomorphic to $V_n$.  Read ``$s$'' as suspension or shift. For a free graded commutative algebra $\RR[V]$ generated by $V$, consider 
the free graded commutative algebra $\RR[sV].$  For $V$ of finite rank in each degree, $\RR[V]$ splits as $\RR[V_{even}] \otimes  \Lambda V_{odd}$ and similarly for  $\RR[sV].$
\end{rem}

Let us now come to the cases important for the following sections. Starting with a Lie algebroid $A$ and a dual Lie algebroid $A^*$, we can apply the parity reversion functor fiberwise. If $e_i$ and $e^i$ denote  basis elements of local frames of $A$ and $A^*$, respectively, we will denote the corresponding anti-commuting versions by 
\eq{
\Pi e^i =\; \xi^i \;, \qquad \Pi e_i =\; \theta_i\;.
}
As explained above, elements in $\Gamma(\wedge^\bullet A^*)$ can be seen as functions on $\Pi A$, thus we choose as local coordinates on $\Pi A$ the set $(x^i,\xi^j)$. Similarly we take $(x^i,\theta_j)$ as coordinates on $\Pi A^*$. As a consequence, derivations on $\Gamma(\wedge^\bullet A^*)$ can now be interpreted as vector fields on $\Pi A$ and similarly for $\Pi A^*$. In the case of the pair of Lie algebroids $(A,A^*),$ we want to give local expressions for $d_A$ and $d_{A^*}$. Let the anchor $a$ on $A$ and the structure constants be given by\footnote{As a standard example, one can take for $A$ the tangent bundle of a Poisson manifold $M$.} 
\eq{
a(e_i) = a^j_i \partial_j \;, \qquad [e_i,e_j]_A = \; f^k_{\,ij} e_k\;,
}
then the vector field corresponding to $d_A$ (which we denote by the same expression) is given by
\eq{
d_A = \; a^j_i(x) \xi^i \partial_j - \frac{1}{2}f^k_{\,ij}(x)\xi^i \xi^j \frac{\partial}{\partial \xi^k}\;.
}

%\trd{purely cosmetic: it looks better to me if we write $f_{ij}^k(x)$ or can tex handle it so k is above i?}

Dually, if we have for the anchor and structure constants on $A^*$:
\eq{
a_*(e^i) = \; a^{ij}\partial_j\;, \qquad [e^i,e^j]_{A^*} = \; Q_k^{\,ij} e^k\;,
}
the dual vector field is given similarly by
\eq{
d_{A^*} = \;  a^{ij}(x)\theta_i \partial_j - \frac{1}{2}Q_k^{\,ij}(x)\theta_i\theta_j\frac{\partial}{\partial \theta_k}\;.
}
%\trd{this and the above half page - excellent! ausgezeichnet!}
For Lie algebroids, $d_A$ and $d_{A^*}$ square to zero or, in terms of the graded commutator of vector fields, $[d_A,d_A] = 0$ and similarly for $d_{A^*}$. Vector fields of this type are called \emph{homological} (following Vaintrob) and therefore we can rephrase the definition of a Lie algebroid in the following form:
\begin{definition}
A vector bundle $A \rightarrow M$ is called \emph{Lie algebroid} if there exists a homological vector field $d_A$ of degree 1 on the supermanifold $\Pi A$. 
\end{definition}
We remark that the structure constants $f$ and $Q$ appearing in the previous homological vector fields have properties similar to the $f$- and $Q$-fluxes in string theory and therefore we use this notation.

\subsubsection{Legendre transform and Lie bialgebroid condition}
Similar to the purely even case, the cotangent bundles of $\Pi A$ and $\Pi A^*$ are symplectic supermanifolds. Denoting the coordinates on $T^*\Pi A$ by $(x^i,\xi^j,x^*_i,\xi^*_j)$, where $*$ denotes the linear dual, the canonical Poisson brackets are given by
\eq{
\{x^i,x^*_j\} = \delta^i_j\;, \quad \{\xi^i,\xi^*_j\} = \delta^i_j\;, 
}
with all the other combinations vanishing. The same statement holds for $T^*\Pi A^*$, whose local coordinates should be denoted by $(x^i, \theta_j, x^*_i, \theta_* ^j)$. The relation between the  $T\Pi A$ and  $T^*\Pi A$ is given by a symplectomorphism, the Legendre transform\footnote{For a detailed motivation, we refer the reader to \cite{kirill&ping} or \cite{Deethesis}.}
 $L : T^*\Pi A \rightarrow T^*\Pi A^*$. It associates the fibre coordinates $\xi^i$ on $\Pi A$ to the conjugate momenta $\theta_*^i$ on $\Pi A^*$, i.e. in local coordinates on $T^*\Pi A$:
\eq{
L(x^i,\xi^j,x^*_i,\xi^*_j) = \; (x^i,\xi_j^*,x^*_i,\xi^j)\;.
}

To get familiar with the Legendre transform, we  represent sections $X \in \Gamma(A),\  \omega \in \Gamma(A^*)$ as functions on $T^*\Pi A$. The situation is given by the following diagram:
\eq{
\begin{matrix}
T^*\Pi A & \overset{L}\rightarrow & T^*\Pi A^* \\
\downarrow \mathfrak{p} & & \downarrow \bar{\mathfrak{p}} \\
\Pi A & & \Pi A^* 
\end{matrix}
}
First, let $X = X^i(x)e_i \in \Gamma(A)$ or, similarly, $Y = Y^i(x)\theta_i$ as a function on  $\Pi A^*$. 

Then the interior derivative $i_X$ as a vector field on $\Pi A$ is given by $i_X = X^i(x)\frac{\partial}{\partial \xi^i}$, which can be written as a function on $T^*\Pi A$ as $h_{i_X} = X^i(x)\xi^*_i$ and we have 
\eq{
L^*\bar{\mathfrak{p}}^* X = L^*\Bigl(Y^i(x)\theta_i\Bigr) = X^i(x)\xi^*_i = h_{i_X}\;.
}
Similarly a section $\omega = \omega_i(x) e^i \in \Gamma(A^*)$ gives rise to the function $\eta_i(x) \xi^i$. The interior derivative $i_\omega$ is given by the vector field $i_\omega = \omega_i(x) \frac{\partial}{\partial \theta_i}$ on $\Pi A^*$ and thus is a function $h_{i_\omega} = \omega_i(x) \theta_* ^i$ on $T^*\Pi A^*$. We have
\eq{
L^* h_{i_\omega} = L^*\Bigl(\omega_i(x) \theta_*^i\Bigr) = \eta_i(x) \xi^i = \mathfrak{p}^* \omega\;.
}
Furthermore it is possible to interpret the differentials $d_A$ and $d_{A^*}$ as functions on $T^*\Pi A$, again with the help of the pullback by $L$ in the case of $d_{A^*}$. In local coordinates we have
\eq{
  h_{d_A} =& \;  a^j_i(x) x_j ^*\xi^i -\frac{1}{2} f^k_{\,ij}(x)\xi^i\xi^j \xi^*_k \\
  L^* h_{d_{A^*}} =&\; a^{ij}(x)x_i^* \xi_j^* - \frac{1}{2} Q_k^{\,ij}(x)\xi^k\xi_i^*\xi_j^*.
}
Thus, we  can define a meaningful sum of the two differentials as functions on $T^*\Pi A$, namely $$\mu = h_{d_A} + L^* h_{d_{A^*}}\;.$$ 

As an application of this rewriting, we state a theorem which gives an elegant characterization of a Lie bialgebroid, a detailed proof being given in \cite{Deethesis}:
\begin{theorem}
A pair of dual Lie algebroids $(A,A^*)$ is a Lie bialgebroid if and only if $ \{\mu,\mu\} = 0 $.
\end{theorem}
This theorem is enough motivation for a definition of the Drinfel'd double of a Lie bialgebroid. 
\begin{definition}
\label{drinfeld}
The \emph{Drinfel'd double} of a Lie bialgebroid $(A,A^*)$ is given by $T^*\Pi A$ together with the homological vector field $D = \{\mu,\cdot\}$.
\end{definition}

\begin{rem}
In \cite{kirill:Crelle}, Mackenzie refers to the   non-super version as the \emph{cotangent double of $(A,A^*)$}, distinguishing it from the \emph{Courant algebroid} $A\oplus A^*$ of \cite{lwx}. The super approach helps to explain some commonality of these various notions.
\end{rem}

We will see in later sections that, in a suitable choice of coordinates, the Drinfel'd double and especially the corresponding homological vector field will play a dominant role in making contact with the structures arising in string theory. Before establishing these results, we will digress into a very basic review of important aspects of double field theory.

\section{Glances at double field theory}
It is now time to pause the exposition of supermathematical language and to review a selection of elements of double field theory. Closed string theory with toroidal target space geometries exhibits two different sets of ``momenta''. On the one hand, there are the standard momenta canonically conjugate to the center of mass coordinates of the string. But in addition there is the possibility of a string winding around compact directions, resulting in winding degrees of freedom. Taking the latter onto a similar footing as the standard momentum, the resulting coordinates conjugate to winding give a second set of coordinates referred to as \emph{winding coordinates}. The physical fields depend on both sets of coordinates.

\begin{rem}
In DFT, it is the coordinates that are doubled and `double fields' refer to fields which depend on both sets of coordinates.
Thus we also refer to double functions, double vector fields, double forms, etc.
\end{rem}

%\trd{reverting to red for this file - the above is the clearest exposition I've seen - well done!}

%\trh{Thanks a lot! Good to hear.}

\subsection{Winding, double fields and the strong constraint}
In order to motivate these concepts, we choose the simplest example of closed string theory with toroidal target space equipped with a metric $G$ and Kalb-Ramond $B$-field (i.e. a two-form on $T^d$), both being constant. 

%\trd{ If both G and B are on the target, omit `geometry' above  Explain what a Kalb-Ramond $B$-field is - an arbitrary 2-from? on $T^d$?}

The corresponding sigma model is given in terms of maps $X : \Sigma \rightarrow T^d$ from a two-dimensional domain $\Sigma = \R\times S^1$ to a $d$-dimensional torus $M = T^d$. Taking coordinates $(\tau,\sigma)$ on the world sheet $\Sigma$, the action reads:
\eq{
\label{sigmamodel}
S = \frac{1}{4\pi} \int_0 ^{2\pi}\,d\sigma \int_0^\infty\, d\tau \Bigl(h^{\alpha\beta}\partial_\alpha X^i \partial_\beta X^j G_{ij} + \epsilon^{\alpha\beta}\partial_\alpha X^i\partial_\beta X^j B_{ij}\Bigr)\;.
}
We use the two-dimensional diagonal metric $h = \textrm{diag}(-1,1)$ and the antisymmetric symbol $\epsilon$ with the convention $\epsilon^{01} = -1$. The latter is used to indicate that this term is a two-form on the world sheet.

%\trd{Is that just to keep track of skew-commutativity?  can say it otherwise?}

 Furthermore $\partial_\alpha = (\partial_\tau, \partial_\sigma)$ and we assume the periodicity $X^i = X^i + 2\pi$ in terms of  the torus coordinates.

The momentum $P_i$ canonically conjugate to the target space coordinate $X^i$ is given by 
\eq{
P_i = \frac{1}{2\pi} \Bigl(G_{ij}\partial_\tau X^j + B_{ij}\partial_\sigma X^j\Bigr)\;,
} 
and the Hamiltonian density corresponding to the Lagrangian in \eqref{sigmamodel} is determined to be 
\eq{
h = \frac{1}{4\pi}\Bigl(\partial_\sigma X, 2\pi P\Bigr){\cal H}(G,B) 
\begin{pmatrix}
\partial_\sigma X \\
2\pi P 
\end{pmatrix}
\;,
}
where we introduced the \emph{generalized metric} $\cal H$, a symmetric bi-linear form on the direct sum $TM \oplus T^*M$, depending on the metric $G$ and $B$-field in the following way:
\eq{
\label{genmetric}
{\cal H}(G,B)  = \begin{pmatrix}
G - BG^{-1}B & BG^{-1} \\
-G^{-1}B & G^{-1} 
\end{pmatrix}\;. 
}

Classical solutions to the equations of motion of the sigma model \eqref{sigmamodel} with periodic boundary conditions are given by the sum of a left- and a right-moving part, 
\eq{
\label{classic}
X_R^i =& x_{0R}^i + \alpha_0^i(\tau - \sigma) + i\underset{n\neq 0}\sum \frac{1}{n} \alpha_n^i e^{-in(\tau - \sigma)} \;,\\
X_L^i =& x_{0L}^i + \bar \alpha_0^i(\tau + \sigma) + i\underset{n\neq 0}\sum \frac{1}{n} \bar \alpha_n^i e^{-in(\tau +\sigma)}\;,
}
with constants $x_{0R},x_{0L}$, oscillator coefficients $\alpha_n^i,\bar \alpha_n^i, \, n\neq 0$ and zero modes $\alpha_0^i, \bar \alpha_0^i$ given in double momentum space by:
\eq{
\alpha_0^i =& \frac{1}{\sqrt{2}}G^{ij}\Bigl(p_j - (G_{jk} + B_{jk})w^k\Bigr)\;,\\
\bar\alpha_0^i =&\frac{1}{\sqrt{2}} G^{ij}\Bigl(p_j + (G_{jk} - B_{jk})w^k\Bigr)\;,
}
where the $w^i$ are defined by $w^i = \frac{1}{2\pi}\int_0 ^{2\pi}\partial_\sigma X^i$. The $p_j$ are interpreted as the canonical momenta similarly to the point particle case, but in addition there is the momentum corresponding to winding, given by $w^k$. The motivation to consider a doubling of coordinates is to take the two sets of momenta as canonically conjugate to two different sets of coordinates, in the sense that for canonical quantization we get the following operators:
\eq{
\label{operators}
p_i \simeq \frac{1}{i}\frac{\partial}{\partial x^i}\;, \quad \textrm{and} \quad w^i \simeq \frac{1}{i}\frac{\partial}{\partial \tilde x_i}\;.
}

%\trd{This is excellent - momenta doubled before coordiantes!!}

Note, that this introduction of a double set of coordinates is a consequence of the appearance of winding degrees of freedom in the classical solution \eqref{classic} of the closed string sigma model. Taking this idea further, classical double fields are real or complex functions (or more generally sections in the appropriate tensor bundles), depending on both sets of coordinates. Note the coordinates are doubled, not the functions, vector fields, forms, etc.

Furthermore, physical fields in DFT are restricted by the \emph{strong constraint}\footnote{We are not going to elaborate on its origin, but refer the reader to the literature, e.g. \cite{Hohm:2010jy}.}. E.g. for scalar fields $\phi(x^i,\tilde x_i)$ and $\psi(x^i,\tilde x_i)$ this constraint is given by 
\eq{
\partial_i\phi (x^i,\tilde x_i)\tilde \partial^i \psi(x^i,\tilde x_i) + \tilde\partial^i \phi(x^i,\tilde x_i) \partial_i \psi(x^i, \tilde x_i) = 0\;,
}
where the $\tilde \partial^i$ are the partial derivatives with respect to the $\tilde x_i$-coordinates, as given in \eqref{operators}.
We will interpret this condition in terms of the Drinfel'd double in later sections. But for the next subsection, we take these ingredients more or less axiomatically and review results of double field theory.

\subsection{Spacetime action and gauge algebra}
The evolution of DFT in the last decade has many faces and different routes. We are not in a position to describe all of its physical motivations coming from areas as different as gauged supergravity or string field theory. Even citing all of the literature is a complicated task and so we refer the reader to the expositions which provided us with a trail through the world of double field theory (\cite{Hull:2009zb,Zwiebach:2011rg,Hull:2009mi}). 

%\trd{well said}

One of the main results achieved so far is to formulate a spacetime action for fields depending on the doubled set of coordinates, having a global $O(d,d)$ symmetry (motivated by the action of \emph{T-duality} at the spacetime level) and reproducing the standard low energy effective string action when reduced to half of the coordinates:
\eq{
S= \int d^d x \sqrt{-\det G}\; e^{-2\Phi}\Bigl(R + 4\partial^i\phi \partial_i\phi - \frac{1}{12}H_{ijk}H^{ijk} \Bigr)\;.
}
%\trd{what is g?}\trh{oh sorry.}
We denote the field strength of the Kalb-Ramond field by $H = dB$, by $\Phi(x)$ the string dilaton %\trd{what's a dilaton? I've avoid learning up to now}
%\trh{It is a scalar field included in the massless spectrum of the closed string. The massless fields are $G,B $ and $\Phi$.}
and by $R$ the Ricci scalar determined by the metric $G_{ij}$. 

The double field action having these properties depends on the doubled dilaton $\mathfrak{D}(x,\tilde x)$ and the generalized metric ${\cal H}_{MN}(G(x,\tilde x),B(x,\tilde x))$\footnote{In the following, we denote pairs of indices corresponding to the doubled  coordinates by capital letters. Mathematically they correspond to  the structure of the direct sum of the tangent- and cotangent bundle of the physical configuration space, but coefficient functions depending on both sets of coordinates. As an example we have  $V^N = \Bigl(V^i(x,\tilde x),V_i(x,\tilde x)\Bigr)$ and $V_N = \Bigl(V_i(x,\tilde x),V^i(x,\tilde x)\Bigr)$.}
%\tgr{something missing in that footnote?}
%\tgr{so how does one write $V_N$ ?}
defined in \eqref{genmetric}, but now also allowed to depend on the doubled set of coordinates. It is given by:
\eq{
\label{dftaction}
S_{{\cal H}} = \int \, dx d\tilde x e^{-2\mathfrak{D}} \Bigl( &\frac{1}{8} {\cal H}^{MN}\partial_M {\cal H}^{KL} \partial_N {\cal H}_{KL} -\frac{1}{2}{\cal H}^{MN}\partial_N{\cal H}^{KL}\partial_L {\cal H}_{MK}  \\
&2\partial_M \mathfrak{D} \partial_N{\cal H}^{MN} + 4{\cal H}^{MN}\partial_M \mathfrak{D} \partial_N \mathfrak{D} \Bigr)\;.
}
\begin{rem}
The notation we use here makes global $O(d,d)$-symmetry manifest. On a $d$-dimensional vector space $V$, $O(d,d)$-transformations are defined by the requirement 
\eq{
A \eta A^t = \eta\;, \quad A \in O(d,d)\;,
}
where $\eta$ is the bilinear form on on $V \oplus V^*$ represented in matrix form by
\eq{
\eta_{MN} = \begin{pmatrix}
0 & \textrm{id} \\
\textrm{id} & 0 
\end{pmatrix}\;.
} 
Indices $M,N$ are raised and lowered by contraction with $\eta_{MN}$ or its inverse $\eta^{MN}$.
\end{rem}

Instead of discussing problems of this action such as  the integration measure or covariance and attempts to provide solutions, we are going to state one additional symmetry \cite{Hull:2009zb} of the action which we are going to interpret mathematically in the following section. Let $\Sigma$ be the pair $(X(x^i,\tilde x_j),\omega(x^i,\tilde x_j))$\footnote{I.e. in components $(\Sigma)^M = (X^i,\eta_j)$.}. We will refer to objects of this kind as ``double vectors'' in the following. 
%\trd{can we think of  better name for $\Sigma$ or a tleast say first : we will refer to the pair $ (X(x^i,\tilde x_j),\omega(x^i,\tilde x_j)$ as a ``double'' vector}
In addition, locally define the \emph{generalized derivative} ${\cal L}_{\Sigma}$ to act on double scalars and double vectors in the following way:
\eq{
{\cal L}_\Sigma \phi =& \Sigma^M \partial_M \phi \;,\\
({\cal L}_\Sigma V)_M =& \Sigma^K\partial_K V_M + (\partial_M\Sigma^K -\partial^K\Sigma_M)V_K \;,\\
({\cal L}_\Sigma W)^M =& \Sigma^K\partial_K W^M -(\partial_K\Sigma^M -\partial^N\Sigma_K)W^K\;,
}

\noindent and extended as a derivation to double fields with more indices, similarly to the standard Lie derivative on vector fields.

%\trd{meaning as a derivation?}

 It turns out that the double field action \eqref{dftaction} is invariant under transformations given by the action of the generalized Lie derivative, meaning for the fields ${\cal H}_{MN}$ and $d$:
\eq{
\delta_\Sigma {\cal H}_{MN} &= ({\cal L}_\Sigma {\cal H})_{MN}\;,\\
\delta_\Sigma e^{-2d} &= \partial_M\Bigl(\Sigma^M e^{-2d}\Bigr)\;,
}
%\tgr{why the exponential - does $\Sigma$ act on d by itself?  just tradition?}
where the latter is the generalization of the transformation of a scalar density to double fields. To get information about the nature of the generalized Lie derivative and the bracket structure underlying the space of double vectors, it is important to consider the commutator of two generalized Lie derivatives. It is given for $\Sigma_1(x^i,\tilde x_j) = (X(x^i,\tilde x_j),\eta(x^i,\tilde x_j))$ and $\Sigma_2(x^i,\tilde x_j) = (Y(x^i,\tilde x_j),\omega(x^i,\tilde x_j))$ by 
\eq{
\Bigl[{\cal L}_{\Sigma_1},{\cal L}_{\Sigma_2}\Bigr] = \; -{\cal L}_{[\Sigma_1,\Sigma_2]_C}\;,
}
where the bracket determining the algebra of ``double sections'' is given by the \emph{C-bracket}. Note that the geometric meaning of ``double sections'' is not obvious at this point in our exposition of double field theory and we only use it as a name for pairs given above. We will provide an interpretation in the next section, which leads to the right bracket structure.

%\trd{``doubled sections''?? of ...??}
%\trh{Absolutely, this is NOT clear in double field theory up to now. We will provide an explanation as functions on $T^*\Pi TM$, depending on a special mix of odd and even coordinates. I indicated this now in addition, see the last sentence above.}

Because the C-bracket will be the starting point of relating the geo\-me\-try of symplectic supermanifolds to the algebraic relations for doubled sections, we state the component form for this bracket in full detail. In $O(d,d)$-notation, it is given by:
\eq{
\Bigl([\Sigma_1,\Sigma_2]_C\Bigr)^M = \Sigma_1^K\partial_K\Sigma_2^M - \Sigma_2^K\partial_K \Sigma_1 ^M -\frac{1}{2}\Bigl(\Sigma_1^K\partial^M\Sigma_{2K} - \Sigma_2 ^K\partial^M\Sigma_{1K}\Bigr)\;,
}
or, if we sort by the various components of different type, we get:
\eq{
\label{Cbracket}
\Bigl([X,Y]_C\Bigr)_i =&\; 0\;, \quad 
\Bigl([X,Y]_C \Bigr)^i =\; X^k\partial_k Y^i - Y^k\partial_k X^i \;, \\
\Bigl([X,\omega]_C\Bigr)_i =&\; \;\;X^k\partial_k\omega_i -\frac{1}{2}(X^k\partial_i\omega_k - \omega_k\partial_i X^k) \;,\\
\Bigl([X,\omega]_C\Bigr)^i =&\;-\omega_k \tilde\partial^k X^i - \frac{1}{2}(X^k\tilde\partial^i\omega_k - \omega_k\tilde\partial^i X^k) \;,\\
\Bigl([\eta, Y]_C \Bigr)_i =&\; -Y^k\partial_k \eta_i + \frac{1}{2}\Bigl(Y^k\tilde\partial^i\eta_k - \eta_k\tilde \partial^i Y^k \Bigr)\;, \\
\Bigl([\eta,Y]_C \Bigr)^i =&\; \;\;\eta_k\tilde \partial^k Y^i + \frac{1}{2}\Bigl(Y^k\tilde\partial^i\eta_k - \eta_k\tilde\partial^i Y^k\Bigr)\;, \\
\Bigl([\eta,\omega]_C\Bigr)_i =& \; \;\;\eta_k\tilde \partial^k \omega_i - \omega_k\tilde \partial^k \eta_i\;, \quad 
\Bigl([\eta,\omega]_C\Bigr)^i = 0\;.
}

Mathematically, it is not clear at this point how to interpret doubled sections, having vector- and form-like parts but depending on a doubled set of coordinates. 

%\trd{meaning they are not really sections? or sections of something as yet undetermined?}
%\trh{Yes. We found: They are functions on $T^*\Pi TM$ depending on a special mix of even and even combination of odd variables, as we define in the next section. I think this is one of our messages to the mathematicians: We uncovered the geometric meaning of ``double vectors'' or ``double sections'' used in physics to get the gauge algebra of double field theory.}

The algebra of gauge transformations parametrized by these objects is governed by the C-bracket and the transformations are a symmetry of the double field spacetime action. In the following we are going to re-interpret double fields as functions on an even symplectic supermanifold and derive the expression for the C-bracket by using its Poisson algebra. 
%\tgr{not sure what you mean - C-bracket gives a Poisson bracket?}

\section{C-bracket and strong constraint from Poisson algebra}
After these mainly introductory sections, we will now choose a specific setting to derive the C-bracket and the strong constraint of double field theory in a purely algebraic way. Starting with a configuration space manifold $M$, we use the even symplectic structure of $T^*\Pi TM$ introduced in section \ref{sec-language}. The main argument to connect to double field theory will be the choice of two different sets of canonical momenta and their conjugate coordinates.

\subsection{Two sets of momenta and the C-bracket}
As we have seen in the previous subsection on double field theory, dynamical fields depend on $2n$ coordinates if the physical configuration space has dimension $n$. To make contact with physical observables, one has to reduce the number of coordinates. In the following, we will start with an $n$-dimensional configuration space $M$ and define two different sets of momenta $p_a$ and $\tilde p^a$ on $T^*\Pi TM$ being related to the original even momenta $x^*_i$ by the addition of an even combination of the $\xi^i$ and $\xi^*_i$. 
%\trd{well said!!}
Thus we implicitly get two different sets of coordinates $x^a, \tilde x_a$ and have achieved a ``formal'' doubling of the configuration space. The justification of this procedure is given by the exact reproduction of the C-bracket and the strong constraint of double field theory. 

The choice of two sets of momenta on $T^*\Pi TM$ is canonically given by the form of the differentials $d_A$ and $d_{A^*}$ introduced in previous sections. For convenience, we now set $A = TM$, and thus $A^* = T^*M$. The definition of the new momenta, with index $a$, is dictated by:
\eq{
h_{d_A} =& \xi^i\Bigl(a^j_i(x) x_j^* -\frac{1}{2}f^k_{\,ij}(x)\xi^j\xi^*_k\Bigr) =:\xi^a p_a\;,\\  
h_{d_{A^*}} =& \xi^*_i\Bigl(a^{ij}x_j^* + \frac{1}{2}Q_k^{\,ij}\xi^k \xi^*_j\Bigr) =: \xi^*_a \tilde p^a\;.
}
Clearly, the two sets $p_a$, $\tilde p^a$ are related to each other, so to get one independent set of variables again, we have to choose one set or a specific combination of both. One way is given by the strong constraint in double field theory and will be the topic of the next subsection. First, we will use the momenta and their canonically conjugate variables to derive the C-bracket in an algebraic way. The conjugate variables $x^a, \tilde x_a$ are defined by the standard Poisson algebraic relations:
\eq{
\{p_a, x^b\} = \delta_a^b\;, \quad \{\tilde p^a, \tilde x_b\} = \delta^a_b\;,
}
and we will label functions $f$ on the configuration space $M$ by $f(x^i) = f(x^a,\tilde x_a)$, which is well-defined because the $x^a,\tilde x_a$ are related to the original coordinates $x^i$ by a coordinate transformation and the possible addition of even combinations of the $\xi^i$ and $\xi^*_i$.  Finally, to state the main result of this section, a derivation of the C-bracket of double field theory in terms of the Poisson algebra on $T^*\Pi A$, we introduce a unified notation for the lift of vector fields and one-forms to $T^*\Pi A$: Let $p : T^*\Pi A \rightarrow \Pi(A \oplus A^*)$ be the projection such that 
\eq{
p^* X = L^* \bar{\mathfrak{p}}^* X = h_{i_X}\; &\textrm{for}\; X\in \Gamma(A)\;,\\
p^* \omega = \mathfrak{p}^* \omega \; &\textrm{for}\; \omega \in \Gamma(A^*)\;.
}
In local coordinates, this means for the lifts of vector fields and one-forms to functions on $T^*\Pi A$:
\eq{
p^*(X^i e_i) = \; X^i \xi^*_i \;, \qquad \; p^*(\omega_j e^j) = \omega_j \xi^j\;.
}
As noticed above, we can view the coordinate functions $X^i,\omega_j$ as being dependent on the two sets of coordinates:\  $X^i(x,\tilde x), \omega_j(x,\tilde x)$. Then one has the following result:
\begin{theorem}
\label{maintheorem}
Let $X+\eta$ and $Y+\omega$ be sections of $A\oplus A^*$ with corresponding lifts to $T^*\Pi A$ given by $\Sigma^1 = p^*(X+\eta)$, $\Sigma^2 = p^*(Y+\omega)$. In addition, let the bilinear operation $\circ$ be defined as
\eq{
\Sigma^1 \circ \Sigma^2 = \; \Bigl\{\{\xi^a p_a + \xi^*_a \tilde p^a, \Sigma^1\},\Sigma^2\Bigr\}\;.
}
Then the C-bracket of $\Sigma^1$ and $\Sigma^2$ in double field theory is given by 
\eq{
[\Sigma^1,\Sigma^2]_C =\; \frac{1}{2}\Bigl(\Sigma^1 \circ \Sigma^2 - \Sigma^2 \circ \Sigma^1 \Bigr)\;.
}
%\tgr{maybe not even worth a remark but this show it to be a `derived' structure - cf. T Vorononv}
\end{theorem}
\begin{proof}
The proof is done by explicit calculation in local coordinates. As this illustrates the use of lifting to functions on $T^*\Pi A$, we list the explicit steps. We first note that
\eq{
[\Sigma^1,\Sigma^2]_C = [p^*X,p^*Y]_C + [p^*X,p^*\omega]_C + [p^*\eta,p^*Y]_C + [p^*\eta,p^*\omega]_C \;,
}
and compute the various terms separately. For the mixed terms we have:
\eq{
p^*X\circ p^*\omega =& \{\{\xi^a p_a + \xi_a^*\tilde p^a, X^i(x,\tilde x)\xi_i^*\}, \omega_j(x,\tilde x) \xi^j\} \\ 
=&\{\partial_a X^i\xi^a\xi_i ^* + X^a p_a +  \tilde \partial^a X^i\xi_a ^* \xi_i^*, \omega_j \xi^j\} \\
=& \omega_i\partial_a X^i\xi^a +  X^a\partial_a \omega_j\xi^j +  \omega_i \tilde\partial^a X^i\xi_a ^*  -  \omega_a \tilde\partial^a X^i\xi_i^* \\
=& (\omega_k\partial_i X^k + X^k\partial_k \omega_i)\xi^i + (\omega_k \tilde \partial^i X^k - \omega_k \tilde \partial^k X^i)\xi^* _i \;.
}

Analogously, we have
\eq{
p^*\omega \circ p^*X =& (X^k\partial_i\omega_k -X^k\partial_k\omega_i)\xi^i + (X^k\tilde\partial^i \omega_k + \omega_k\tilde\partial^k X^i)\xi^*_i\;.
}
Thus we get for the C-bracket:
\eq{
[p^*X,p^*\omega]_C = &\frac{1}{2}(p^*X\circ p^*\omega - p^*\omega \circ p^*X) \\
=& \;\Bigl(X^k\partial_k\omega_i -\frac{1}{2}(X^k\partial_i\omega_k - \omega_k\partial_i X^k)\Bigr)\xi^i \\
&+ \Bigl(-\omega_k \tilde\partial^k X^i - \frac{1}{2}(X^k\tilde\partial^i\omega_k - \omega_k\tilde\partial^i X^k)\Bigr)\xi_i^*\;.
}  
A similar calculation leads to $[p^*\eta,p^*Y]_C$. By using the same algebra as before, we can also compute the parts consisting of pure vector fields and pure one-forms: 
\eq{
[p^*X,p^*Y]_C =& \; (X^k\partial_k Y^i - Y^k\partial_k X^i)\xi_i^*\;, \\
[p^*\eta,p^*\omega]_C =& \;(\eta_k\tilde \partial^k \omega_i - \omega_k\tilde \partial^k \eta_i)\xi^i\;.
}
which is manifestly the C-bracket \eqref{Cbracket} encountered in the section on double field theory. 
\end{proof}
\begin{rem}
This form of the C-bracket in terms of a two-fold Poisson bracket shows that it is a derived structure in the sense of \cite{MR2223157,zbMATH02217997}.
\end{rem} 

\subsection{Drinfel'd double and the strong constraint}
Having established an interpretation of the gauge algebra of double field theory in terms of the symplectic structure on $T^* \Pi TM$, it is natural to ask about the meaning of the conditions defining the Drinfel'd double as given in definition \ref{drinfeld}. Recall, that for $\mu = h_{d_A} + L^*h_{d_{A^*}}$, as consequence of $\{\mu,\mu\}=0$, the operator $D =\{\mu, \cdot\}$ is a homological vector field. Additionally in the last section we introduced momenta $p_a$ and $\tilde p^a$ to rewrite $\mu$ as $\mu = p_a\xi^a + \tilde p^a\xi^*_a$. 

Let $\phi(x^i,\tilde x_i)$ and $\psi(x^i,\tilde x_i)$ be two functions on $T^*\Pi TM$ depending on the coordinates $x^i$ and $\tilde x_i$. Using the homological vector field $D$, we get 
\eq{
0 =&\; D^2\phi = \Bigl\{\mu,\{\mu,\phi\}\Bigr\} \\
=&\;\Bigl\{p_b \xi^b + \tilde p^b \xi^*_b, \{p_a \xi^a + \tilde p^a \xi^*_a, \phi\}\Bigr\} \\
=&\;\Bigl\{ p_b \xi^b + \tilde p^b \xi^*_b, \partial_a \phi \,\xi^a + \tilde \partial^a\phi\, \xi^*_a \Bigr\}\\
=&\; p_a \tilde \partial^a \phi + \tilde p^a \partial_a \phi\;.
}
Taking furthermore the bracket with the second function $\psi(x^i,\tilde x_i)$, we get
\eq{
0 = &\; \{D^2\phi, \psi\} \\
=&\;\{ \tilde p^a \partial_a \phi + p_a \tilde \partial^a \phi, \psi\} \\
=&\; \partial_a \phi \,\tilde \partial^a \psi + \tilde \partial^a \phi\, \partial_a \psi\;.
}
Following the interpretation of double scalar fields as functions on $T^*\Pi TM$, depending on $x^a$ and $\tilde x_a$, we recover the \emph{strong constraint} of double field theory. As this is an essential constraint that physical fields in DFT should obey, we can also say that the gauge algebraic structure of the latter is governed by the Drinfel'd double of the Lie bialgebroid $(A=TM,A^*=T^*M)$. In the following section we will see how solutions to the strong constraint result into the familiar Courant brackets used in DFT and in Poisson geometry.

\begin{rem}
We remark that there are also different constraints for physical fields used in the literature on DFT. We leave it for future work to interpret also these types of conditions in terms of symplectic geometry on $T^*\Pi TM$.
\end{rem}

\subsection{Projection to half of the coordinates}
As we have seen, the strong constraint can be interpreted for functions on $T^*\Pi TM$ as the Drinfel'd double condition for a Lie bialgebroid. In DFT, physical fields have to obey the strong constraint. Different solutions to the latter correspond to different physical ``frames'' (sometimes called polarizations), related by $O(d,d)$ transformations, if $M$ is $d$-dimensional. In the following we are going to interpret the simplest solution to the strong constraint in terms of the geometric language of the last sections. 

To begin with, let us recall the following observation related to the C-bracket of double field theory (e.g. \cite{Hull:2009zb}) for the simplest solution of the strong constraint: Considering only fields depending on $x^i$ (i.e. setting all the derivatives $\tilde \partial^i$ to zero), the C-bracket reduces to the original Courant bracket of \cite{zbMATH00004959}, i.e. for pairs $(X^i(x),\eta_j(x))$ and $(Y^i(x),\omega_j(x))$, the C-bracket reduces to:
\eq{
\label{oldCourant}
[X+\eta,Y+\omega]_{\textrm{Cour}} := \; [X,Y] + L_X \omega -L_Y \eta + \frac{1}{2}d(i_Y\eta -i_X\omega)\;.
}
From the viewpoint of the last sections, dropping the dependence on $\tilde x_i$ is equivalent to setting the canonical momenta conjugate to these coordinates to zero, i.e. $\tilde p^i = 0$. We will now show that this reproduces the bracket \eqref{oldCourant} by doing a similar calculation as in the proof of theorem \ref{maintheorem}. Let $\Sigma_1 = X^m(x)\xi_m^* + \eta_m(x) \xi^m$ and $\Sigma_2 = Y^k(x)\xi_k^* + \omega_k(x) \xi^k$ be the lifts of the above pairs to functions on $T^*\Pi TM$. Then we have:
\eq{
\Sigma_1 \circ \Sigma_2 = & \Bigl\{\{p_a \xi^a, X^m\xi^*_m + \eta_m \xi^m\},Y^k\xi^*_k + \omega_k \xi^k\Bigr\} \\
=& \Bigl\{p_aX^a + \partial_a X^m\xi^a\xi^*_m + \partial_a \eta_m \xi^a\xi^m, Y^k\xi_k^* + \omega_k \xi^k\Bigr\} \\
= & (X^k\partial_k Y^i - Y^k\partial_k X^i)\xi^*_i  \\
& +(X^k\partial_k\omega_i + \omega_k \partial_i X^k + Y^k\partial_i \eta_k -Y^k\partial_k \eta_i)\xi^i \;.
}
Antisymmetrization of this expression gives the result 
\eq{
\frac{1}{2}(\Sigma_1 \circ \Sigma_2 - \Sigma_2 \circ \Sigma_1)=&\; (X^k\partial_k Y^i -Y^k\partial_k X^i)\xi^*_i \\
&+ \Bigl( X^k\partial_k \omega_i - X^k\partial_i\omega_k + \frac{1}{2}\partial_i(X^k\omega_k) \\
&-Y^k\partial_k \eta_i + Y^k\partial_i\eta_k - \frac{1}{2}\partial_i(Y^k\eta_k)\Bigr)\xi^i\;,
}
which is just the bracket \eqref{oldCourant}, lifted to a function on $T^*\Pi TM$. As a result, we are able to reproduce this case of reducing to physical fields in terms of symplectic geometry of the Drinfel'd double.

\section{Conclusion and outlook}
Even symplectic supermanifolds give the possibility to interpret mathematically the notion of a double field and the gauge algebra of DFT. The C-bracket is given by Poisson brackets and the strong constraint is a consequence of the condition determining the Drinfel'd double of a Lie bialgebroid.

These results are a sufficient motivation for continuing this super\-ma\-the\-ma\-ti\-cal viewpoint on DFT. Still restricting to the gauge algebra of the latter, an open physical question is the appropriate inclusion of three-forms and skew-symmetric three-tensor fields (which are the ``fluxes'' in string theory) into the algebra. The advantage of the supermathematical language is the fact that there is an immediate possibility to achieve this: Let $(A,A^*)$ be a pair of Lie algebroids in duality, $H \in \Gamma(\wedge^3 A^*)$ a three-form on $A$ and $R\in \Gamma(\wedge^3 A)$ a three-vector on $A$. Consider now the lifts of these quantities to $T^*\Pi A$ by using in addition the Legendre transform for the vectors. If, according to \cite{Deethesis}, one defines:
\eq{
\nu = h_{d_A} + L^*h_{d_{A^*}} + \mathfrak{p}^* H + L^*\bar{\mathfrak{p}}^* R \;,
}
%\tgr{I have a feeling p is overused, but can't think of a better subscript}
a \emph{proto-bialgebroid} is given by the supermanifold $T^*\Pi A$ together with the condition that $\{\nu,\nu\} = 0$. Appropriate Dorfman-type brackets are formulated in \cite{Deethesis}. An important question for future work is to compare this to the physics of double fields and the appearance of $H$- and $R$-fluxes there, especially the right inclusion of the parts depending on the winding coordinates. Furthermore the relation of the condition on $\nu$ to the Bianchi-identities of \cite{Blumenhagen:2012pc} has to be explored.

\vspace{40pt}

\emph{Acknowledgements:}  The authors want to thank Dimitry Roytenberg and Kirill Mackenzie for providing them with background knowledge and for discussions. Furthermore, we thank Barton Zwiebach and Ping Xu for discussions and the organisers of the workshop ``Non-associativity in mathematics and physics''\footnote{\url{www.gravity.psu.edu/events/math_phys/program.shtml}} at the Pennsylvania State University for providing a stimulating atmosphere.

\bibliographystyle{alpha}
\bibliography{mybig}
\end{document}